\documentclass[a4paper]{article}

\usepackage{INTERSPEECH2022,amsmath,graphicx}
\usepackage{amssymb}
\usepackage{mathtools}
\usepackage{cite}

\usepackage{amsthm}
\theoremstyle{definition}
\newtheorem{theorem}{Theorem}

\newtheorem{prop}[theorem]{Proposition}

\title{How bad are artifacts?:\\Analyzing the impact of speech enhancement errors on ASR}
\name{Kazuma Iwamoto$^{1}$, Tsubasa Ochiai$^{2,\star}$, Marc Delcroix$^{2}$, \\ Rintaro Ikeshita$^{2}$, Hiroshi Sato$^{2}$, Shoko Araki$^{2}$, Shigeru Katagiri$^{1}$}
\address{
  $^1$Doshisha University, $^2$NTT Corporation\thanks{$^{\star}$ Corresponding author. Email: tsubasa.ochiai.ah@hco.ntt.co.jp}}

\begin{document}
\maketitle
\begin{abstract}

It is challenging to improve automatic speech recognition (ASR) performance in noisy conditions with single-channel speech enhancement (SE).
In this paper, we investigate the causes of ASR performance degradation by decomposing the SE errors using orthogonal projection-based decomposition (OPD).
OPD decomposes the SE errors into noise and artifact components.
The artifact component is defined as the SE error signal that cannot be represented as a linear combination of speech and noise sources.
We propose manually scaling the error components to analyze their impact on ASR.
We experimentally identify the artifact component as the main cause of performance degradation, and we find that mitigating the artifact can greatly improve ASR performance.
Furthermore, we demonstrate that the simple observation adding (OA) technique (i.e., adding a scaled version of the observed signal to the enhanced speech) can monotonically increase the signal-to-artifact ratio under a mild condition.
Accordingly, we experimentally confirm that OA improves ASR performance for both simulated and real recordings.
The findings of this paper provide a better understanding of the influence of SE errors on ASR and open the door to future research on novel approaches for designing effective single-channel SE front-ends for ASR.

\end{abstract}

\section{Introduction}
\label{sec:introduction}

Building automatic speech recognition (ASR) systems that are robust to acoustic interference such as background noise and reverberation remains a challenging problem in speech processing.
Many studies have shown that multichannel speech enhancement (SE) front-ends, e.g., mask-based neural beamformers \cite{heymann2016neural,erdogan2016improved,boeddeker2018exploring}, can greatly improve ASR robustness \cite{barker2015third,barker2018fifth}.
However, in many practical situations, it is not always possible to use a microphone array.
Therefore, developing single-channel SE approaches for robust ASR is still an important research topic. 

Most single-channel SE approaches tend to only slightly improve or even degrade ASR performance \cite{yoshioka2015ntt,chen2018building,menne2019investigation,fujimoto2019one}.
Retraining the ASR back-end on enhanced speech \cite{chen2018building} and joint training of the SE front-end and ASR back-end \cite{menne2019investigation} have been investigated to improve the ASR performance of single-channel SE.
However, it is not always possible to modify the ASR back-end.
Moreover, the performance improvements are limited and much smaller than when using a microphone array.
How to design single-channel SE front-ends that can significantly improve ASR performance remains an open research question.

In general, it is often assumed that processing distortions induced by the single-channel SE are the cause of poor ASR performance.
However, there have been no detailed systematic analyses or interpretations of such distortions, particularly of their impact on ASR.
We argue that a deeper understanding of the impact of single-channel SE estimation errors on ASR is essential for designing better SE front-ends.

To this end, we propose a novel analysis by decomposing the SE estimation errors using orthogonal projections, which was previously proposed for speech enhancement and separation performance metrics \cite{vincent2006performance}.
Orthogonal projection-based error decomposition decomposes the SE errors into two components, i.e., the {\it noise} and {\it artifact} components ($\mathbf{e}_{\text{noise}}$ and $\mathbf{e}_{\text{artif}}$ in Figure~\ref{fig:decomposition}-(a)).
These two components are obtained by projecting the SE errors onto (1) a speech-noise subspace spanned by the speech and noise signals and (2) a subspace orthogonal to the speech-noise subspace.

The noise component consists of a linear combination of speech and noise signals, and thus it represents a signal that could naturally be observed.
We denote these as {\it natural} signals.
The effect of such natural signals on the ASR performance could be limited since similar noise components would naturally appear in the training samples.
On the other hand, the artifact component consists of a signal that cannot be represented by a linear combination of speech and noise signals, and thus it represents an artificial and {\it unnatural} signal.
Such unnatural signals may be very diverse and rarely appear in the training samples.
Therefore, we hypothesize that ASR would be more sensitive to the artifact components than to the noise components\footnote{The above consideration is motivated by the fact that the enhanced signal of the beamformer, which is known to significantly improve ASR, is a linear combination of the filtered multichannel observation signals, which may be considered a {\it natural} signal.}.

In this paper, we design a novel analysis scheme based on the orthogonal projection-based error decomposition to confirm our above hypotheses and better understand the impact of SE errors on ASR.
We propose manually controlling the ratio of the noise and artifact errors in the enhanced signals and experimentally observe the impact of each error component on the ASR performance.
With this analysis, we are able to prove the particularly negative impact of artifact errors on ASR performance as well as the relatively limited impact of the noise errors.

In addition, we explore the effectiveness of the observation adding (OA) technique as the post-processing for improving the ASR performance.
OA simply consists of adding back a scaled version of the observed signal to the enhanced speech.
It has been used as a practical approach to mitigating the auditory impact of processing distortions.
We propose an interpretation of this simple technique from the viewpoint of orthogonal projection analysis, which justifies its effectiveness in improving the ASR performance.

The experimental analyses proposed in this paper provide a novel insight on the causes of the limited ASR performance improvement induced by single-channel SE.
The results show that single-channel SE approaches that were known to be ineffective as ASR's front-end could have the potential to improve ASR performance by more carefully considering the artifacts.

\section{Error decomposition of estimate sources}

In this paper, we focus on the single-channel SE (noise reduction) task.
Let $\mathbf{y} \in \mathbb{R}^{T}$ denote an $T$-length time-domain waveform of the observed signal.
The observation signal is modeled as $\mathbf{y} = \mathbf{s} + \mathbf{n}$, where $\mathbf{s} \in \mathbb{R}^{T}$ denotes the source signal and $\mathbf{n} \in \mathbb{R}^{T}$ denotes the background noise signal.

The aim of SE is to reduce the noise signal from the observed signal.
Given the observed signal $\mathbf{y}$ as an input, the enhanced signal $\widehat{\mathbf{s}} \in \mathbb{R}^{T}$ is estimated as $\widehat{\mathbf{s}} = \text{SE}(\mathbf{y})$, where $\text{SE}(\cdot)$ denotes the SE processing carried out, e.g., by a neural network.

\subsection{Orthogonal projection-based decomposition}
\label{sec:decomposition}

Let us review the orthogonal projection-based error decomposition approach that was first introduced for designing SE evaluation metrics.
The estimated signal $\widehat{\mathbf{s}}$ inevitably contains estimation errors.
Vincent et al. \cite{vincent2006performance} proposed decomposition of the estimated signal in three terms as\footnote{Unlike Vincent et al. \cite{vincent2006performance}, this paper focuses on the single-channel SE (noise reduction) task without interfering speakers, and thus we reformulate the equations by focusing on the noise and artifact errors.}:
\begin{align}
    \widehat{\mathbf{s}} = \mathbf{s}_{\text{target}} + \mathbf{e}_{\text{noise}} + \mathbf{e}_{\text{artif}},
    \label{eq:decomp}
\end{align}
where $\mathbf{s}_{\text{target}} \in \mathbb{R}^{T}$ is the target source component, and $\mathbf{e}_{\text{noise}} \in \mathbb{R}^{T}$  and $\mathbf{e}_{\text{artif}} \in \mathbb{R}^{T}$ denote the noise and artifact error components, respectively.

We can derive this decomposition using orthogonal projections.
Let $\mathbf{s}^{\tau}$ and $\mathbf{n}^{\tau}$ be the source and noise signals delayed by $\tau$.
We denote by $\mathbf{P}_{\mathbf{s}} \in \mathbb{R}^{T \times T}$ the orthogonal projection matrix onto the subspace spanned by the source signals $\{ \mathbf{s}^{\tau} \}_{\tau = 0}^{L-1}$, where $L-1$ is the number of maximum delay allowed.
Similarly, $\mathbf{P}_{\mathbf{s},\mathbf{n}} \in \mathbb{R}^{T \times T}$ is the orthogonal projection matrix onto the subspace spanned by the source and noise signals $\{ \mathbf{s}^{\tau}, \mathbf{n}^{\tau} \}_{\tau = 0}^{L-1}$.
These matrices can be obtained as follows:
\begin{align}
    \mathbf{P}_{\mathbf{s}} &\coloneqq \mathbf{A}_{\mathbf{s}}(\mathbf{A}_{\mathbf{s}}^{\mathsf{T}} \mathbf{A}_{\mathbf{s}})^{-1} \mathbf{A}_{\mathbf{s}}^{\mathsf{T}}, \\
    \mathbf{P}_{\mathbf{s},\mathbf{n}} &\coloneqq \mathbf{A}_{\mathbf{s}, \mathbf{n}}(\mathbf{A}_{\mathbf{s}, \mathbf{n}}^{\mathsf{T}} \mathbf{A}_{\mathbf{s}, \mathbf{n}})^{-1} \mathbf{A}_{\mathbf{s}, \mathbf{n}}^{\mathsf{T}},
\end{align}
where $\mathbf{A}_{\mathbf{s}} \coloneqq [ \mathbf{s}^{\tau=0}, \ldots, \mathbf{s}^{\tau=L-1} ] \in \mathbb{R}^{T \times L}$ and $\mathbf{A}_{\mathbf{s}, \mathbf{n}} \coloneqq [ \mathbf{s}^{\tau=0}, \ldots, \mathbf{s}^{\tau=L-1}, \mathbf{n}^{\tau=0}, \ldots, \mathbf{n}^{\tau=L-1} ] \in \mathbb{R}^{T \times 2L}$.

The decomposed terms in Eq.~\eqref{eq:decomp} can be obtained using the projection matrices as:
\begin{align}
    \mathbf{s}_{\text{target}} &= \mathbf{P}_{\mathbf{s}} \widehat{\mathbf{s}},\label{eq:target} \\
    \mathbf{e}_{\text{noise}} &= \mathbf{P}_{\mathbf{s},\mathbf{n}} \widehat{\mathbf{s}} - \mathbf{P}_{\mathbf{s}} \widehat{\mathbf{s}}, \label{eq:noise} \\
    \mathbf{e}_{\text{artif}} &= \widehat{\mathbf{s}} - \mathbf{P}_{\mathbf{s},\mathbf{n}} \widehat{\mathbf{s}}. \label{eq:artif}
\end{align}
Figure~\ref{fig:decomposition}-(a) illustrates the signal decomposition, where, for illustration simplicity, we assume no signal delay, i.e., $L=1$.

The above decomposition was originally proposed to derive the SE evaluation metrics \cite{vincent2006performance}, such as 1) signal-to-distortion ratio (SDR), 2) signal-to-noise ratio (SNR), and 3) signal-to-artifact ratio (SAR).

\begin{figure}[t]
  \centering
  \includegraphics[width=0.8\linewidth]{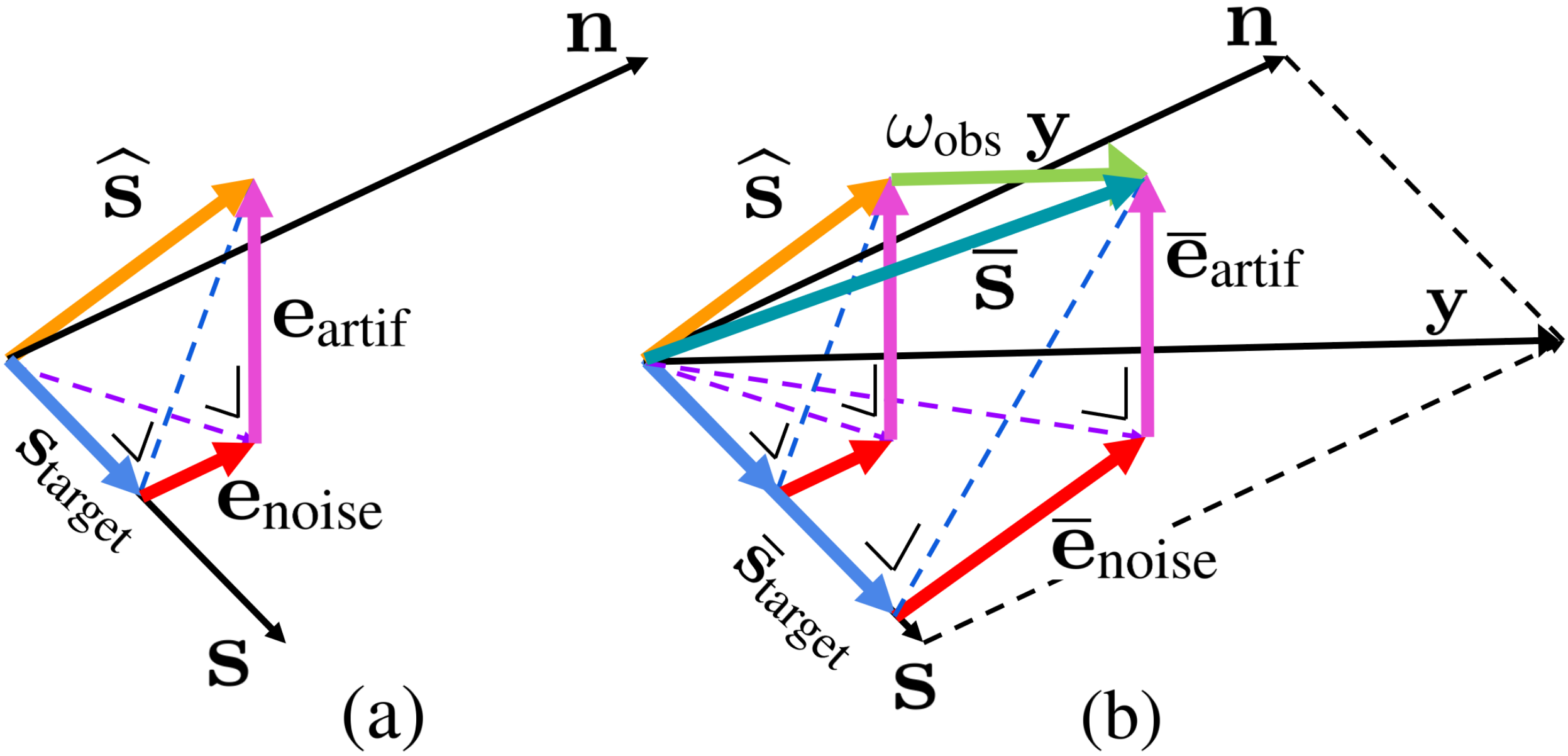}
\vspace{-3mm}
  \caption{Illustrations of (a) orthogonal projection-based decomposition and (b) effect of observation adding.}
\vspace{-3mm}
  \label{fig:decomposition}
\end{figure}

\section{Analysis schemes}
\label{sec:analysis}

From Eq.~\eqref{eq:artif}, the artifact error component $\mathbf{e}_{\text{artif}}$ represents an {\it unnatural} signal generated by the SE systems, which cannot be expressed as the linear combination of the {\it natural} signals (i.e., source and noise signals).
We hypothesize that such unnatural error components would have bad effects on the ASR performance.
Below in Section~\ref{sec:dsa}, we propose an experimental scheme to investigate the relative impact of the error components on ASR.
In Section~\ref{sec:oa}, we also give an interpretation to the effect of the OA technique, which can mitigate the impact of the artifact errors under a mild condition.

\subsection{Controlling SNR and SAR by directly scaling error components}
\label{sec:dsa}

To measure the impact of the different error components, we propose modifying enhanced signals by varying the scale of the error components.

After decomposing the estimated signal with the orthogonal projection-based decomposition described in Section \ref{sec:decomposition}, we synthesized a modified version of the enhanced signal $\widehat{\mathbf{s}}_{\omega} \in \mathbb{R}^{T}$ by directly scaling (increasing/decreasing) the error components $\mathbf{e}_{\text{noise}}, \mathbf{e}_{\text{artif}}$ as:
\begin{align}
    \widehat{\mathbf{s}}_{\omega} &= \mathbf{s}_{\text{target}} + \omega_{\text{noise}}\ \mathbf{e}_{\text{noise}} + \omega_{\text{artif}}\ \mathbf{e}_{\text{artif}},\label{eq:dsa}
\end{align}
where $\omega_{\text{noise}}$ and $\omega_{\text{artif}}$ are the parameters that control the amount of the noise and artifact error components.

We changed the values of $\omega_{\text{noise}}$ and $\omega_{\text{artif}}$ to obtain various modified enhanced signals $\widehat{\mathbf{s}}_{\omega}$. This allows us to control the SNR and SAR values while retaining the same target source component $\mathbf{s}_{\text{target}}$.
By inputting such modified enhanced signals to the ASR systems, we can directly measure the impact of each error component on the ASR performance.
In the following experiment, we refer to this modification approach as direct scaling analysis (DSA).
Note that DSA requires having access to the reference source and noise signals, and thus it cannot be applied for real recordings.

\subsection{Improving SAR with observation adding technique}
\label{sec:oa}

OA simply consists of adding back a scaled version of the observed signal to the enhanced speech.
The modified enhanced signal $\overline{\mathbf{s}} \in \mathbb{R}^{T}$ is computed as:
\begin{align}
    \overline{\mathbf{s}} &= \widehat{\mathbf{s}} + \omega_{\text{obs}}\ \mathbf{y} = \widehat{\mathbf{s}} + \omega_{\text{obs}}\ \mathbf{s} + \omega_{\text{obs}}\ \mathbf{n},\label{eq:oa}
\end{align}
where $\omega_{\text{obs}} \geq 0$ is a parameter that controls the amount of the added observation signal.
Figure~\ref{fig:decomposition}-(b) illustrates the effect of observation adding in modifying the target source and noise error components, but not adding artifacts ($\overline{\mathbf{e}}_{\text{artif}} = \mathbf{e}_{\text{artif}}$).
Unlike DSA, it does not require having access to the reference speech and noise signals, and thus can be applied for real recordings.

OA is a simple technique often used to reduce the auditory nonlinear distortion (i.e., musical noise \cite{cappe1994elimination}) of the enhanced signals generated by binary masking \cite{lyon1983computational} and spectral subtraction \cite{boll1979suppression}.
However, to the best of our knowledge, such a technique has not been investigated as an approach to improve the ASR performance of recent single-channel SE systems such as neural network-based ones.
Moreover, its effect on the decomposed estimation errors has not been considered.
In this paper, we show that the OA technique has a potential to mitigate the impact of the artifact errors, which could contribute to improving the ASR performance based on our hypothesis.

In Section~\ref{sec:experiment2}, we experimentally observe that
the OA technique improves the SAR of the enhanced signal $\widehat{\mathbf{s}} = \mathrm{SE}(\mathbf{y})$ for all the evaluated utterances, when using a time-domain SE neural network that is trained with the phase-aware scale-dependent loss function.
This empirical observation can also be explained analytically as follows.
According to Proposition~\ref{prop:OA}, which we provide below, the OA technique improves the SAR of $\widehat{\mathbf{s}} = \mathrm{SE}(\mathbf{y})$ if the mild condition $\langle \widehat{\mathbf{s}}, \mathbf{y} \rangle > 0$ holds.

This condition should be met even for SE front-ends that would not be very accurate, and even more likely to be met for the recent high-quality (phase-aware) SE systems.
Assuming such SE systems, we can roughly expect $\widehat{\mathbf{s}} \approx \mathbf{s} + \varepsilon \mathbf{n} + \widehat{\mathbf{e}}_{\text{artif}}$ with $\varepsilon \in \mathbb{R}$, where $|\varepsilon|$ is small.
In this case,
by $\langle \mathbf{y}, \widehat{\mathbf{e}}_{\text{artif}} \rangle = 0$
and under the natural assumption $\langle \mathbf{s}, \mathbf{n} \rangle \approx \mathbf{0}$,
we have
$
\langle \widehat{\mathbf{s}}, \mathbf{y} \rangle
\approx 
\langle \mathbf{s} + \varepsilon \mathbf{n}, \mathbf{s} + \mathbf{n} \rangle
\approx 
\| \mathbf{s} \|^2 + \varepsilon \| \mathbf{n} \|^2
$.
Since $|\varepsilon|$ is small, we can expect  
$\langle \widehat{\mathbf{s}}, \mathbf{y} \rangle > 0$
and that OA almost always improves SAR.

\begin{prop}
\label{prop:OA}
The OA technique improves the SAR of the original enhanced signal $\widehat{\mathbf{s}} = \mathrm{SE}(\mathbf{y})$ if 
it satisfies $\langle \widehat{\mathbf{s}}, \mathbf{y} \rangle > 0$.
\end{prop}
\begin{proof}
Let $\overline{\mathbf{e}}_{\text{artif}}$ and $\widehat{\mathbf{e}}_{\text{artif}}$ be the artifact errors corresponding to the modified enhanced signal $\overline{\mathbf{s}}$ given by Eq.~\eqref{eq:oa} and original enhanced signal $\widehat{\mathbf{s}} = \mathrm{SE}(\mathbf{y})$, respectively.
Based on Eqs.~\eqref{eq:target}, \eqref{eq:noise}, \eqref{eq:artif}, \eqref{eq:oa}, and the definition of the SAR \cite{vincent2006performance} (i.e., $\text{SAR} \coloneqq 10 \log_{10} \frac{\| \mathbf{s}_{\text{target}} + \mathbf{e}_{\text{noise}} \|^{2}}{\| \mathbf{e}_{\text{artif}} \|^{2}}$), the SAR improvement (denoted as $\text{SARi}$) can be computed as:
\begin{align}
    \nonumber
    \text{SARi} 
    &= 
        10 \log_{10} \frac{\| \mathbf{P}_{\mathbf{s}, \mathbf{n}} \overline{\mathbf{s}} \|^{2}}{\| \overline{\mathbf{e}}_{\text{artif}} \|^{2}} - 10 \log_{10} \frac{\| \mathbf{P}_{\mathbf{s}, \mathbf{n}} \widehat{\mathbf{s}} \|^{2}}{\| \widehat{\mathbf{e}}_{\text{artif}} \|^{2}},
    \\
    \nonumber
    & = 
        10 \log_{10} \frac{\| \mathbf{P}_{\mathbf{s}, \mathbf{n}} \widehat{\mathbf{s}} + \omega_{\text{obs}} \mathbf{y} \|^{2}}{\| \mathbf{P}_{\mathbf{s}, \mathbf{n}} \widehat{\mathbf{s}} \|^{2}},
    \\
    \nonumber
    &= 
        10 \log_{10}
        \left[ 
            1
            + 
            \frac{
                \omega_{\text{obs}}^{2} \| \mathbf{y} \|^2 + 2 \omega_{\text{obs}} \langle \mathbf{P}_{\mathbf{s}, \mathbf{n}} \widehat{\mathbf{s}}, \mathbf{y} \rangle
            }{
                \| \mathbf{P}_{\mathbf{s}, \mathbf{n}} \widehat{\mathbf{s}} \|^2
            }
        \right],
\end{align}
where we used 
$\mathbf{P}_{\mathbf{s}, \mathbf{n}} \mathbf{y} = \mathbf{y}$
and
$\overline{\mathbf{e}}_{\text{artif}} = \widehat{\mathbf{e}}_{\text{artif}}$
in the second equality.
Thus,
$\text{SARi} > 0$ holds if
$\langle \mathbf{P}_{\mathbf{s}, \mathbf{n}} \widehat{\mathbf{s}}, \mathbf{y} \rangle > 0$.
This sufficient condition can be rewritten as
$
\langle \mathbf{P}_{\mathbf{s}, \mathbf{n}} \widehat{\mathbf{s}}, \mathbf{y} \rangle
= 
\langle \widehat{\mathbf{s}}, \mathbf{P}_{\mathbf{s}, \mathbf{n}} \mathbf{y} \rangle
=
\langle \widehat{\mathbf{s}}, \mathbf{y} \rangle  > 0
$,
which concludes the proof.
\end{proof}

\section{Experiments}
\label{sec:experiment}

\subsection{Conditions}

\textbf{Evaluated Data:}
We created a dataset of simulated reverberant noisy speech signals from the Wall Street Journal (WSJ0) corpus for speech source \cite{garofalo2007csr} and the CHiME-3 corpus for noise source \cite{barker2015third}.
To create the reverberant speech source, we randomly generated simulated room impulse responses based on the image method \cite{allen1979image}, where we set the reverberation time (T60) between 0.2 s and 0.6 s.

We created 30,000, 5,000, and 5,000 noisy speech signals for training, development, and evaluation sets, respectively.
The input SNR of the training and development sets was randomly selected for each utterance between 0 dB and 10 dB, while the input SNR of the evaluation set was set to 0 dB.
The speech sources for training, development, and evaluation sets were selected from WSJ0’s training set ``si\_tr\_ s,'' development set ``si\_dt\_05'' and evaluation set ``si\_et\_05,’' respectively.
The noise data of the CHiME-3 corpus was divided into 3 subsets for training (80 \%), development (10 \%), and evaluation (10 \%).
We used the fifth-channel signals of CHiME-3's noise data for generating noisy speech signals.
The above simulated dataset\footnote{We chose this setup instead of CHiME-3 official dataset to be able to extend the data to multi-speaker conditions in our future works.} is used to analyze the relation between the SE metrics (i.e., SDR, SNR, SAR \cite{vincent2006performance}) and the ASR metric (i.e., word error rate (WER)).
In addition, we also used the real-recorded speech data of the CHiME-3 dataset (``et05\_real'') to confirm the results on real recordings.

\textbf{Speech enhancement front-end:}
We adopted the neural network-based time-domain denoising network, i.e., Denoising-TasNet \cite{kinoshita2020improving}, as the SE front-end for the evaluation.
Our implementation is based on the Asteroid's Conv-TasNet implementation \cite{luo2019conv,pariente2020asteroid}.
Unlike the original study \cite{luo2019conv}, we adopted the classical (phase-aware and scale-dependent) SNR loss \cite{roux2019sdr} to retain the phase and scale information of the target source component.
By following the notations of \cite{luo2019conv}, we set the hyperparameters as: $N=512$, $L=16$, $B=128$, $H=512$, $P=3$, $X=8$, and $R=3$.
We adopted the Adam algorithm \cite{kingma2015adam}, with an initial learning rate of 0.001, and the gradient clipping \cite{pascanu2013difficulty}.
We stopped the training procedure after 100 epochs.

\textbf{Speech recognition back-end:}
To evaluate ASR performance, we created a deep neural network-hidden Markov model (DNN-HMM) hybrid ASR system based on Kaldi’s standard recipe \cite{hinton2012deep,povey2011kaldi}. 
The system was trained using the lattice-free maximum mutual information (MMI) criterion \cite{povey2016purely} with the reverberant noisy speech in the training set.
We used a trigram language model for decoding.

As described in Section~\ref{sec:introduction}, it is not always possible to modify the ASR back-end due to the need to use an already deployed ASR back-end, the cost of training and maintaining an ASR system for different application scenarios, etc.
Thus, this paper focuses on experiments using ASR system trained independently of the SE front-end, i.e., trained with reverberant noisy speech by following the multi-condition training strategy \cite{vincent2017analysis}.

\vspace{-1mm}
\subsection{Experiment 1: Evaluation with direct scaling analysis}

\begin{figure}[t]
  \centering
  \includegraphics[width=0.9\linewidth]{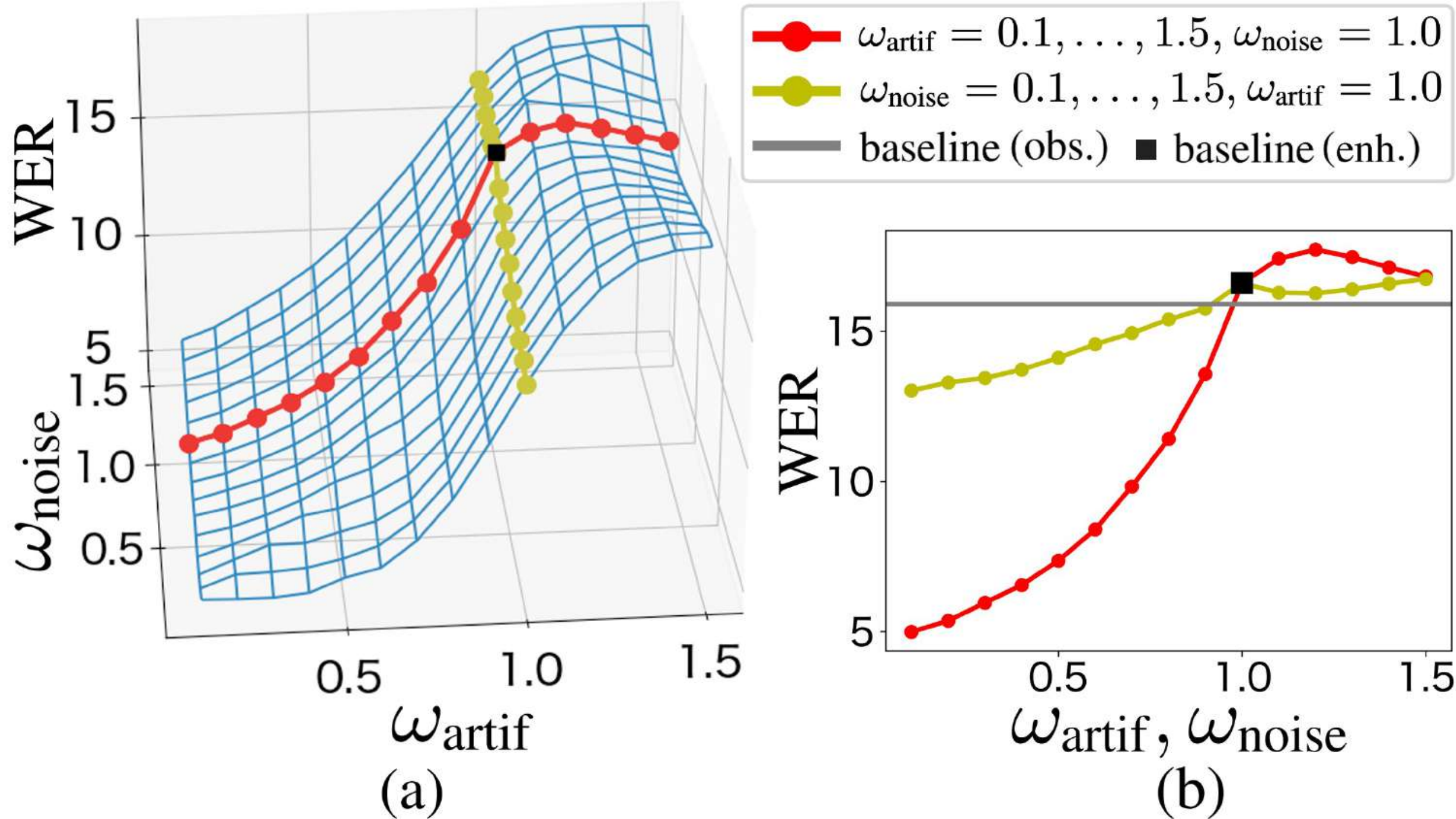}
\vspace{-3mm}
  \caption{Results of DSA approach (WER [\%] (lower is better)).}\label{fig:results_of_DSA}
\vspace{-3mm}
\end{figure}

Figure~\ref{fig:results_of_DSA}-(a) is a 3D plot of the WER scores for the evaluated enhanced signals modified by directly scaling the ratio of the noise and artifact errors based on the DSA method described in Section~\ref{sec:dsa}.
The results were obtained by varying the values of $\omega_{\text{noise}}$ and $\omega_{\text{artif}}$ in Eq.~\eqref{eq:dsa} between $0.0$ and $1.5$. 
Figure~\ref{fig:results_of_DSA}-(b) is the corresponding 2D plot obtained by varying only one of the $\omega_{\text{noise}}$ and $\omega_{\text{artif}}$ weights.
Each curve in the 2D plot corresponds to the curve of the same color in the 3D plot.
In all the figures of this paper, the gray line represents the baseline score of the observed signals, and the square symbol represents the baseline score of the original enhanced signals.

From the figure, we confirm that the original enhanced signals do not improve the ASR performance even compared to the observed signals.
This is similar to what has been reported in previous studies \cite{yoshioka2015ntt,chen2018building,menne2019investigation}.
In addition, Figure~\ref{fig:results_of_DSA} shows that with DSA, it is possible to greatly improve ASR performance by reducing the artifact error component.
In contrast, ASR performance is affected much less by scaling the noise error component.
These results confirm our hypothesis that, among the two types of SE errors, the artifact errors have a larger impact on the degradation of ASR performance.

\subsection{Experiment 2: Evaluation with observation adding}
\label{sec:experiment2}

\begin{figure}[t]
  \centering
  \includegraphics[width=0.95\linewidth]{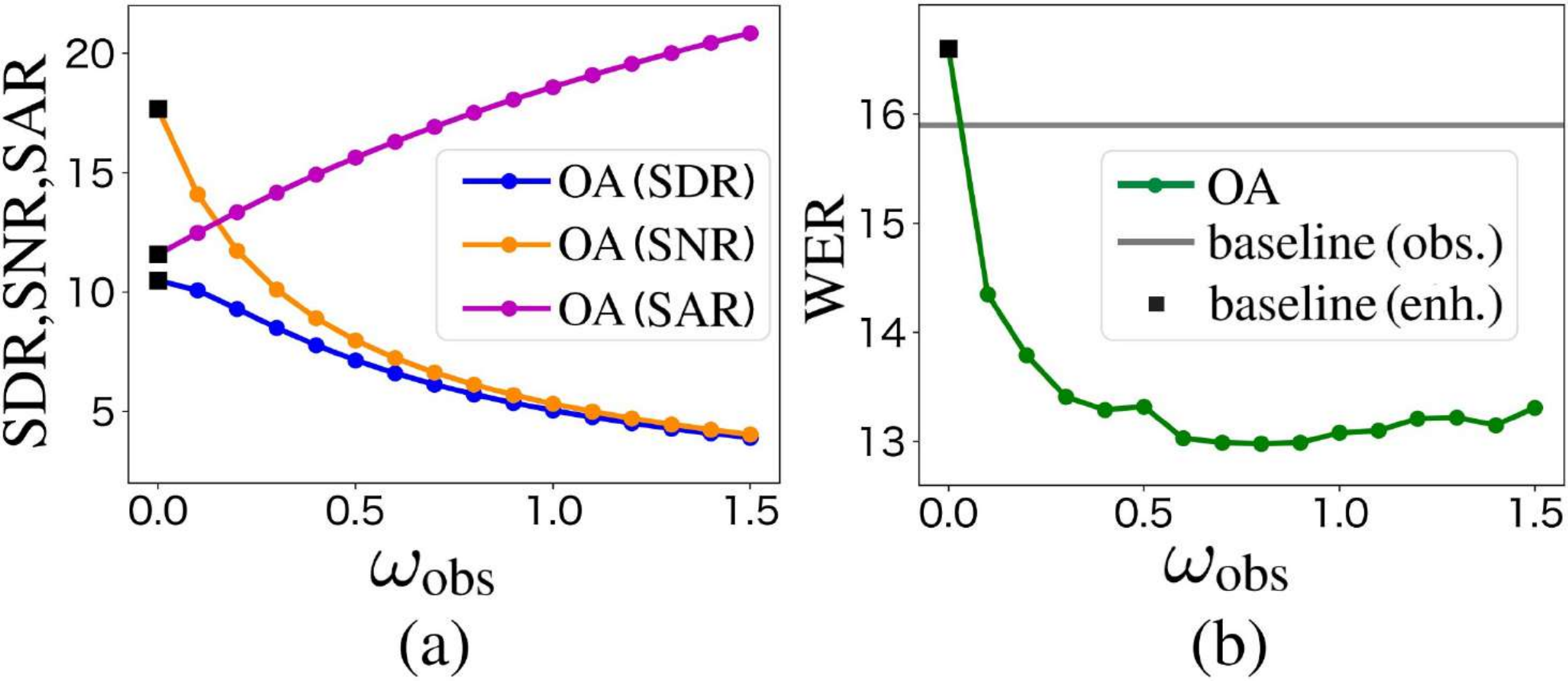}
\vspace{-3mm}
  \caption{Results of OA approach (SDR, SNR, SAR [dB] (higher is better) and WER [\%] (lower is better)).}\label{fig:results_of_OA}
\vspace{-2mm}
\end{figure}

\begin{figure}[t]
  \centering
  \includegraphics[width=0.9\linewidth]{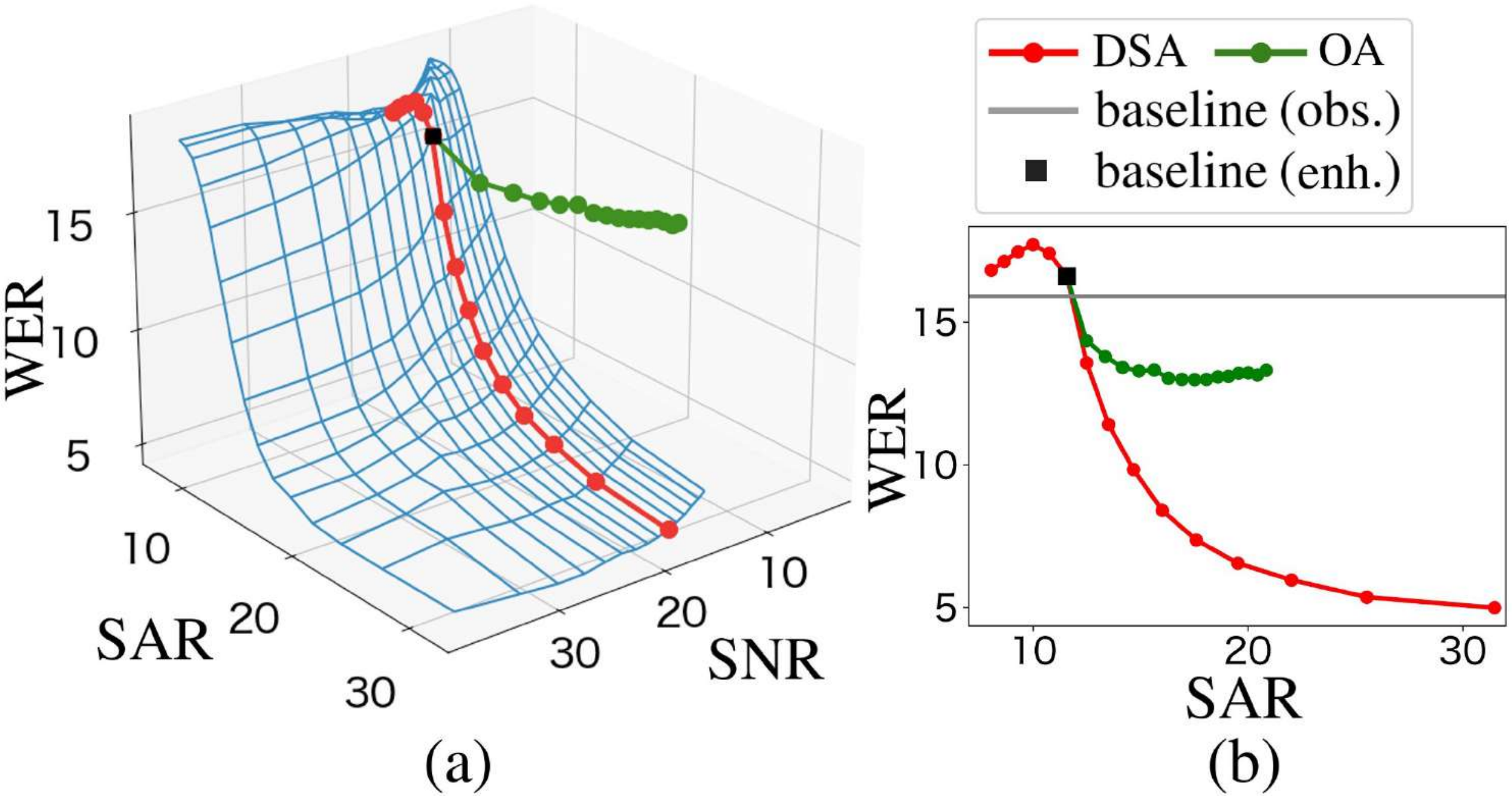}
\vspace{-3mm}
  \caption{Comparison of DSA and OA approaches (WER [\%]).}\label{fig:comparison}
\vspace{-3mm}
\end{figure}

Figure~\ref{fig:results_of_OA} shows the SDR, SNR, SAR, and WER scores for the enhanced signals modified with the OA method described in Section~\ref{sec:oa}. 
The results were obtained by varying the values of $\omega_{\text{obs}}$ in Eq.~\eqref{eq:oa} between 0.0 and 1.5.
Each square symbol ($\omega_{\text{obs}} = 0.0$) represents the score of the original enhanced signals without OA.

The figure shows that as $\omega_{\text{obs}}$ increases (i.e., we add more of the observation signal), the SDR and SNR decrease while the SAR monotonically increases.
Moreover, we observe that the WER improves in conjunction with the SAR improvement\footnote{We confirm that as one would expect, the WER further degrades (approaches that of observed signals) when we increase the ratio of the observation signal by setting $\omega_{\text{obs}}$ to larger values.}.
With the OA method, we can outperform both the baseline observed signals and original enhanced signals.
These results support again our hypothesis that the artifact component has a worse impact on the ASR performance compared to the noise component.
By decreasing the ratio of the artifact errors (i.e., increasing SAR), we could improve the ASR performance of single-channel SE front-ends.

\subsection{Experiment 3: Comparison of DSA and OA}

Both DSA and OA methods described in Section~\ref{sec:dsa} and \ref{sec:oa} altered the SNR and SAR of the evaluated signals by modifying the original enhanced signals.
Figure~\ref{fig:comparison} plots the WER as a function of the SNR and SAR for both experiments.
The red curve corresponds to the result with the DSA method ($\omega_{\text{noise}} = 1.0$), while the green curve corresponds to that with the OA method.

From the figure, we can again confirm that both methods improve the SAR scores and accordingly improve the WER scores.
The result of the OA method (i.e., green curve) shows a smaller WER improvement compared to that of the DSA method (i.e., red curve).
This is probably because the OA method successfully increased the SAR of the evaluated utterances while it decreased the SNR in return for the SAR improvement, as seen in Figures~\ref{fig:results_of_OA}-(a) and \ref{fig:comparison}-(a).

We confirm that the OA method is a practical and simple method that can reduce the negative impact of artifacts on ASR performance.
However, to further improve the ASR performance for the single-channel SE front-ends, it would be worth while to investigate alternative training schemes and objectives, which focus on improving the SAR without degrading the SNR.

\subsection{Experiment 4: Evaluation with real recordings}

\begin{figure}[t]
  \centering
    \includegraphics[width=0.77\linewidth]{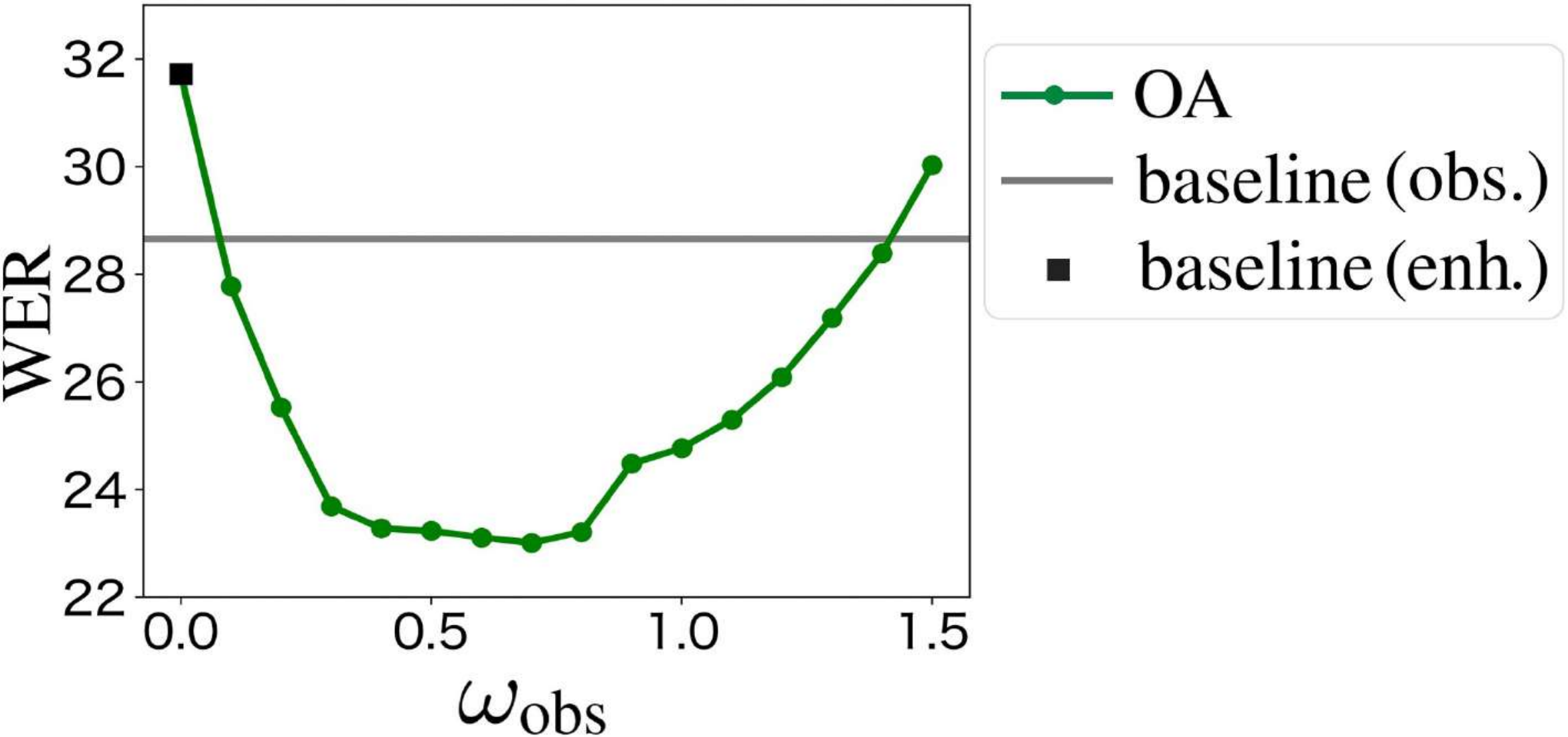}
\vspace{-3mm}
  \caption{Result of OA approach for real recordings (WER [\%]).} \label{fig:real_recording}
\vspace{-3.5mm}
\end{figure}

Finally, we evaluated the effectiveness of the OA method for real recordings.
Figure~\ref{fig:real_recording} shows the WER scores for the enhanced signals modified with the OA method.
We observed that the OA method, which was shown to improve SAR in Section~\ref{sec:experiment2}, is also effective at reducing WER when applied to real recordings.
This suggests that the findings of the error analysis performed on simulated data, i.e., the importance of mitigating artifact errors, would hold for real recordings.

In addition, we can observe that the OA method was not overly sensitive to the value of $\omega_{\text{obs}}$, since it could provide about
20 \% relative WER improvement for scaling $\omega_{\text{obs}}$ in the range of [0.3, 0.8] on the popular CHiME-3 real-recording test set.
This would attest the effectiveness of the OA method in real context.

\vspace{-1mm}
\section{Conclusion}

In this paper, we analyzed the causes of the limited improvement in ASR performance induced by single-channel SE.
We performed orthogonal projection-based decomposition of SE errors and experimentally found that the impact of noise errors is relatively limited, while the impact of artifact errors is particularly detrimental to ASR.

We investigated the OA method as a simple approach to mitigating the influence of artifact errors on ASR.
We demonstrated that OA can monotonically increase the SAR value under a mild condition.
Furthermore, we experimentally confirmed our interpretation of the OA method and its effectiveness in improving the ASR performance, even for real recordings.

We believe that the theoretical interpretations and experimental findings of this paper will give a novel insight into the issues of a single-channel SE front-end for ASR and open novel research directions for single-channel robust ASR.

In our future work, we plan to investigate alternative SE training schemes that consider the importance of minimizing artifacts.
Furthermore, we will extend our experiments to other SE front-ends and ASR back-ends trained on enhanced data.

\vfill
\pagebreak

\bibliographystyle{IEEEtran}
\bibliography{mybib}

\end{document}